\documentclass[reprint,aps,prl]{revtex4-1}
\usepackage{graphicx}
\usepackage{amsmath}
\usepackage{amssymb}
\usepackage{verbatim}
\usepackage{pgf}
\usepackage{bbold}
\usepackage{amsthm}
\usepackage[english]{babel}
\usepackage[all]{xy}

\renewcommand{\r}{\hat\rho}
\newcommand{\T}{\mathcal T}
\newcommand{\U}{\mathcal U}
\renewcommand{\O}{\mathcal O}
\newcommand{\F}{\mathcal F}

\newcommand{\D}{\mathcal D}

\newcommand{\s}{\hat\sigma}

\newtheorem{theorem}{Theorem}

\begin{document}

\title{Dealing with unknown quantum operations}
\author{Filippo Miatto}
\email{fmiatto@uottawa.ca}
\affiliation{SUPA and Department of Physics, University of Strathclyde, Glasgow G4 0NG, Scotland, U.K.}
\affiliation{Current address: Dept. of Physics, University of Ottawa, 150 Louis Pasteur, Ottawa, Ontario, K1N 6N5 Canada}
\date{\today}
\begin{abstract}
In the context of quantum communications between two parties (here Alice and Bob), Bob's lack of knowledge about the communications channel can affect the purity of the states that he receives. The operation of applying an unknown unitary transformation to a state, thus reducing its purity, is called ``twirling''. As twirling affects the states that Bob receives, it also affects his perception of the operations that Alice applies to her states. In this work we find that not every operation is representable after a twirling, we show the minimal requirement for this to be possible, and we identify the correct form of the ``twirled'' operations.
\end{abstract}
\maketitle

\section{Introduction}
Whenever we lack information about which state a quantum system is in, its description is that of a mixed state. Analogously, if we do not know what channel a state is traveling in, at the output it will be mixed. In this context, the operation that decreases the purity of an input state is called ``twirling''.

Twirling finds interesting applications in quantum information, from entanglement distillation \cite{Bennett1996,Bennett1996b}, to the study of entanglement measures \cite{Vollbrecht, Lee, Hayashi}, distillability \cite{Dur}, to particular techniques such as secret sharing \cite{Vlad2012} and others \cite{gour2012, marvian, BRST}.
An analogous way of looking at the same problem is in the context of reference frames \cite{BRS, GS}. 

An interesting consequence of twirling is that states that appear different for the sender might be indistinguishable for the receiver. For this reason, twirling and superselection (the framework that describes the inability of measuring some states and operations) are deeply intertwined.

In this Letter we study how an operation applied by the sender is perceived on the receiving side, given any knowledge (even partial), of the channel connecting the two parties.

\section{Twirling of states}
The framework that we consider is the following: Alice prepares a state $\r$ and sends it to Bob. The unitary channels that she can employ, $\U_g$, are parametrized by elements $g$ of a group $G$. We call the state that Bob receives $\s$, which can eventually be equal to $\r$. Bob, not knowing which of all the possible channels $\r$ went through, is forced to describe $\s$ as a mixture, according to every possible channel:
\begin{align}
\s =\int_{G}dg\, w(g)\,\U_g[\r]:=\T_w[\r],
\label{definition}
\end{align}
where $w(g)$ is the probability density over the channels, and it is this function that encodes Bob's knowledge about which channel Alice might be using.
We indicate this operation by $\T_w$, with the subscript reminding us that it depends on the probability density $w(g)$.

We have two limiting cases: the case of complete knowledge, where $w(g)=\delta(g-\tilde g)$ for some $\tilde g\in G$, and the case of minimal knowledge, where $w(g)=1$. In these limiting cases Bob receives either the pure state $\T_\delta[\r]=\U_{\tilde g}[\r]$, or the completely incoherent mixture $\T[\r]=\int_G\U_{g}[\r] \,dg$. In the general case of partial knowledge, Bob receives a partially mixed state, effectively obtained from the initial state by multiplying its off-diagonal entries by numbers of modulus not larger than 1: the representation of $G$ on the state space, being unitary, induces a Fourier transform, or a Fourier series, of $w(g)$ on the manifold that underlies $G$.

As a simple example consider $G=U(1)$, the twirling of a density matrix in the eigenbasis of some selfadjoint operator $\hat L$ is
\begin{align}
\T_w[\r]&=\frac{1}{2\pi}\sum_{m,n}r_{m,n}\int_0^{2\pi}d\phi\, w(\phi) \,e^{i\phi\hat L}\,|m\rangle\langle n| \,e^{-i\phi\hat L}\nonumber\\
&=\frac{1}{2\pi}\sum_{m,n}r_{m,n}\int_0^{2\pi}d\phi\, w(\phi) \,e^{i\phi(m-n)}|m\rangle\langle n|\nonumber\\
&=\sum_{m,n}r_{m,n}\tilde w_{m-n}|m\rangle\langle n|.
\label{fourier}
\end{align}
The numbers $\tilde w_{m-n}$ are the Fourier coefficients of $w(g)$, so only the diagonal is always safe from alteration, as $w_0=1$. The net effect is that in general the purity of the state that Alice sends, which initially is $\sum_{m,n}|r_{m,n}|^2$ becomes $\sum_{m,n}|r_{m,n}|^2|\tilde w_{m-n}|^2$, so as $|\tilde w_{m-n}|^2\leq1$, it cannot increase, and in general it is reduced.

We show now some properties satisfied by twirling.
The first property is that sequential twirlings can be composed. The second property is that twirling behaves in an associative way, i.e. $(\T_{w''}\T_{w'})\T_w=\T_{w''}(\T_{w'}\T_w)$. To avoid clutter in the equations, we omit the symbol ``$\circ$'', with the rule that the writing $\mathcal A\mathcal B$ means ``$\mathcal A$ after $\mathcal B$''.

Let's consider the composition of two twirlings $\T_w\T_{w'}$ and check if it still corresponds to a twirling $\T_v$ for some probability distribution $v(g)$. Using the properties of the Haar measure, it is easy to see that it is the case, by explicitly calculating $v(g)$:
\begin{align}
(\T_{w'}\T_w)[\r]&=\int_{G}dg\,dh\, w'(h)w(g)\,\U_{hg}[\r]\\
&=\int_{G}dg\,dh\,\,w'(h)w(h^{-1}g)\U_{g}[\r] .
\end{align}
So the composition of the two initial twirlings is equivalent to a twirling with respect to the probability distribution
\begin{align}
v(g)=\int_G\,dh\,w'(h)w(h^{-1}g):=(w'*w)(g).
\label{newprob}
\end{align}
Eq. \eqref{newprob} is the definition of convolution on a group, and it always yields an allowed probability distribution. Therefore we have the composition rule
\begin{align}
\T_w\T_{w'}=\T_{w*w'},
\label{composition}
\end{align}
where the symbol ``$*$'' means convolution on the group $G$.
Now we can easily  prove associativity:
\begin{align}
(\T_{w''}\T_{w'})\T_w&=\T_{w''*w'}\T_w=\T_{(w''*w')*w}\\
&=\T_{w''*(w'*w)}=\T_{w''}\T_{w'*w}\\
&=\T_{w''}(\T_{w'}\T_w),
\label{assoc}
\end{align}
where we used the composition rule \eqref{newprob} and the fact that convolution is associative, so as a consequence also twirling is.

\section{Twirling of operations}
Alice now modifies a copy of her state via a quantum operation $\O$ before sending it to Bob: $\hat \rho'=\O[\r ]$, so Bob receives the states $\s=\T_w[\r]$ and $\s'=(\T_w\O)[\r]$:
\begin{align}
\label{diagram1}
\xymatrix@R=3em@C=3em{
\r\ar[r]^\O\ar[d]_{\T_w}&\r'\ar[d]^{\T_w}\\
\s\ar[r]_{\O'}&\s'
}
\end{align}

The natural question to ask is: what is the operation $\O'$ (if there is any) that transforms $\s $ into $\s'$? In other terms, how does Bob perceive Alice's operation $\O$? This question is equivalent to the problem of finding a map $\O'$ that, given $\O$ and $\T_w$, makes the diagram commute: $\O'\T_w=\T_w\O$.

If Bob knew the channel employed, $w(g)$ would be a delta function centred on a specific $\tilde g$, and it would be one of the limiting cases, $\T_\delta=\U_{\tilde g}$. In this case applying $\O'$ would be equivalent to applying the following sequence: $\U_{\tilde g^{-1}}$ to turn $\s $ into $\r $, then  $\O$ to obtain $\r '$ and finally $\U_{\tilde g}$ to obtain $\s '$. Put simply,
\begin{align}
\O' = T_\delta\O\T_\delta^{-1}=\U_{\tilde g}\O\U_{\tilde g^{-1}}.
\end{align}
This is rather trivial, as $\T_\delta$, corresponding to a unitary transformation, is invertible.
The real issue that we want to explore is the assignment $\O\mapsto\O'$ for a general probability distribution $w(g)$.

To go through the discussion, we employ a decomposition of Alice's operation into four parts, exploiting the matrix representation of the density operator. A quantum operation can be written in terms of Kraus operators:
\begin{align}
\O[\r]=\sum_i K_i\r K_i^\dagger.
\end{align}
The linearity of the Kraus representation allows us to split the density matrix into a diagonal matrix, $\hat d$, and a vanishing-diagonal matrix $\hat f$ via two operations $\D$ and $\F$ (for instance $\left(\begin{smallmatrix} a&b\\ c&d \end{smallmatrix}\right)=\left(\begin{smallmatrix} a&0\\ 0&d \end{smallmatrix}\right)+\left(\begin{smallmatrix} 0&b\\ c&0 \end{smallmatrix}\right)$). Equivalently, we have two contributions to the final result:
\begin{align}
\label{decomp}
\O=\O\D+\O\F,
\end{align}
where $\O\D[\r]$ and $\O\F[\r]$ produce, respectively, the matrices $m_1$ and $m_2$, such that $m_1+m_2=\O(\D+\F)[\r]=\O[\r]=\r'$. Note that in general neither of $m_1$ or $m_2$ are proper quantum states, but this is not of our concern as their sum is always the proper quantum state at the output of $\O$.

The reason for this decomposition is that it allows us to study the effect of twirling on the two contributions separately: 
\begin{align}
\label{twodiagrams}
&\xymatrix@R=.5em@C=2em{
&m_1\ar[dd]^{\T_w}\\
\hat d\ar[ur]^{\O}\ar[dr]_{\O'}\ar@(ul,dl)[]_{\T_w}&&\\
&\T_w[m_1]
}
\hspace{0em}
\xymatrix@R=2em@C=2em{
\hat f\ar[r]^{\O}\ar[d]_{\T_w}&m_2\ar[d]^{\T_w}\\
 \T_w[\hat f]\ar[r]_{\O'}&\T_w[m_2]
}
\end{align}
A diagonal matrix is left unchanged under any twirling, i.e. $\T_w\D=\D$, so we can represent the twirling of a diagonal matrix with an arrow on itself.

\subsection{The $\O\mapsto\O'$ assignment}
With reference to the first diagram in \eqref{twodiagrams}, obtained by applying the map $\D$ to the initial state $\r$, we immediately obtain the commutativity relation $\O'\T_w\D=\T_w\O\D$, with no additional restrictions on $\r$ or on $\O$.
The second diagram, instead, obtained by applying the map $\F$ to the initial state $\r$, induces a restriction on the kind of operations $\O$ that allow commutativity: here the twirling has to be at least partially invertible.

If we consider the twirling action on the entries of a density matrix, eq. \eqref{fourier}, we see that we can retrieve individual off-diagonal entries where $\tilde w_{m-n}\neq0$ by entry-wise multiplication by the inverse $\tilde w_{m-n}^{-1}$. We refer to the action of inverting the twirling where $\tilde w_{m-n}\neq0$ as \emph{partial twirling inversion} and we indicate it as $\tilde \T_w^{-1}$. In a non-ideal situation, the condition should be generalised to $|\tilde w_{m-n}|\geq b$ where $b$ is a lower bound determined by the experimental conditions. Note that twirling is generally not invertible, and that partial twirling inversion is not a positive operation.

From the points raised above, we see that the second diagram in \eqref{twodiagrams} commutes if the final result $\T_w[m_2]$ is independent of the entries of $\hat f$ that are irremediably lost after the twirling $\T_w$. We can annihilate only those entries by applying the map $(\tilde\T_w^{-1}\T_w)$, and create the requirement for commutativity: $\T_w\O\F=\T_w\O(\tilde\T_w^{-1}\T_w)\F$. Comparison with the commutativity expression $\O'\T_w\F=\T_w\O\F$, immediately gives $\O'=\T_w\O\tilde\T_w^{-1}$. We can confirm that commutativity of the first diagram does not imply any restriction, as $\T_w\O\D=\T_w\O(\tilde\T_w^{-1}\T_w)\D$ is always satisfied. This means that the commutativity requirement of the second diagram is equivalent to $\T_w\O=\T_w\O(\tilde\T_w^{-1}\T_w)$.

The other consequence is that the commutativity of both diagrams in \eqref{twodiagrams}, together, imply commutativity of the diagram \eqref{diagram1}: $\O'\T_w\D+\O'\T_w\F=\T_w\O\D+\T_w\O\F\Rightarrow\O'\T_w=\T_w\O$, as it should be, and our proof is complete. We are now ready to state the following theorem:

\begin{theorem}
With reference to the definitions of twirling $\T_w$ and of its partial inversion $\tilde \T_w^{-1}$, the operations $\O$ that satisfy $\T_w\O=\T_w\O(\tilde \T_w^{-1}\T_w)$ make the diagram \eqref{diagram1} commutative, with
\begin{align}
\O'=\T_w\O\tilde \T_w^{-1}.
\end{align}
\end{theorem}
\begin{proof}
See above discussion.
\end{proof}

We point out that the matter of distinguishability of operations is an interesting and closely related issue: the assignment $\O\mapsto\O'$ in general is not injective, so there might well be a many-to-one assignment, in which case Bob identifies operations that are different for Alice. This is an issue that has to be addressed in matters regarding superselection.

Another issue that we should mention regards the state-dependence of the commutativity condition: one could just consider states that already lack enough coherence so that partial inversion of twirling is enough to retrieve the initial state. We have that for those states $(\tilde \T_w^{-1}\T_w)=\mathrm{id}$, so the representability condition is always satisfied, which means that whatever $\O$ Alice chooses, there will be an allowed choice of $\O'$ for Bob. However, one should be careful, as the states that make this possible can be quite dull (an example could be randomly rotating a state to make it rotationally invariant). On the other hand, particular instances of this issue might still be interesting: for instance if $w(g)$ is such that twirling affects coherences between particular sets of eigenstates, employing states that do not exhibit such coherences can be beneficial, as they can still maintain coherence between all the other eigenstates, in a decoherence-free subspace fashion.

\subsection{Example}
As an example, we work out the representation of an operation in case of minimal knowledge $w(g)=1$, inducing a total twirling $\T$ and compare it to a previous result \cite{BRS}. In this extreme case, all Fourier coefficients are zero apart from the zeroth one (which acts on the diagonal elements and its value is 1), so $\T$ is the least locally invertible twirling, and it's equivalent to $\D$, so we can conveniently replace any occurrence of $\D$ with $\T$. Commutativity of the first diagram in \eqref{twodiagrams} means $\O'\T=\T\O\T$. Then, notice that as $(\tilde \T^{-1}\T)=\T$, the commutativity requirement is $\T\O=\T\O\T$, so the commutativity of the first diagram implies commutativity of the total diagram \eqref{diagram1}: $\O'\T=\T\O\T=\T\O$. Finally, we have $\O'=\T\O=\T\O\T$, because total twirling is not invertible.

This result is exactly the same that would be reached by applying the prescription in \cite{BRS,GS}. In that reference, in case of total twirling, an operation $\O$ is prescribed to be mapped to
\begin{align}
\O' = \int_G dg\,\U_g\O\U_{g^{-1}}.
\label{twirlOp}
\end{align}
However, we aim now at showing that this prescription is not as universal as it was intended. Considering the following diagram
\begin{align}
\label{diagram4}
\xymatrix@R=2em@C=1em{
\r\ar[rr]^\O\ar[d]_{\T}&&\r'\ar[d]^{\T}\\
\s\ar[rd]_{\O'}&&\s'\\
&\s''&
}
\end{align}
we can look for the requirement that makes $\s''=\s'$, i.e. that $(\O'\T)[\r]=(\T\O)[\r]$:
\begin{align}
\s''&=(\O'\T)[\r]=\int_Gdg\,dh\,\U_g\O\U_{g^{-1}h}[\r]\nonumber\\
&=\int_Gdg\,\U_g\O\T[\r]=(\T\O\T)[\r].
\end{align}
So $\s'=\s''$ and the diagram commutes if $\T\O\T=\T\O$. To no surprise, the commutativity condition is the same that we found by applying the theorem to the example.
However, in the reference \cite{BRS}, the authors omit it, and one is lead to believe that any operation $\O$ that Alice performs always has a corresponding operation $\O'$ for Bob.

It is even more instructive to explicitly show a counterexample to the operation assignment \eqref{twirlOp}. Consider a qubit state, and a total twirling generated by rotations around the $\hat z$ axis. The operation that we use as counterexample needs to violate the requirement $\T\O\T=\T\O$. An example is a rotation about $\hat x$ (of an angle $\theta\neq k\pi$). With reference to diagram \eqref{diagram4}, we have (remember that $\T$ annihilates the off-diagonal values)
\begin{align}
\s'&=\T [e^{-i\frac{\theta}{2}\s_x}\r \,e^{i\frac{\theta}{2}\s_x}]\\
&=\frac{1}{2}\begin{pmatrix}1+(2p-1)\cos\theta-2\Im(b)\sin\theta&0\\\hspace{-3em}0&\hspace{-5em}1-(2p-1)\cos\theta+2\Im(b)\sin\theta\end{pmatrix}\nonumber
\end{align}
for some angle $\theta$. Instead, if we proceed thorugh the diagram the other way, we obtain
\begin{align}
\s''&=\O'\T\biggl[\begin{pmatrix}p&b\\b^*&1-p\end{pmatrix}\biggr]=\O'\biggl[\begin{pmatrix}p&0\\0&1-p\end{pmatrix}\biggr]\nonumber\\
&=\frac{1}{2}\begin{pmatrix}1+(2p-1)\cos\theta&0\\0&1-(2p-1)\cos\theta\end{pmatrix}
\end{align}
So the prescription \eqref{twirlOp} fails in this case, i.e. $\s'\neq\s''$, because the operation that we chose is not representable after the twirling, as it does not satisfy the commutativity requirement. We also see that a general rotation about $\hat x$ is representable after a total $\hat z$-twirling only for states that have $\Im(b)=0$, and this is an example of state-dependent representability.

\section{Conclusion}
We have shown the composition rule and the associativity property of generalised twirling. Then, we have found the requirement for an operation $\O$ to be faithfully mapped to an operation $\O'$ under a generalised twirling. 
We have also shown that the results in \cite{BRS,GS} which inspired our study, lack a way of telling which operations can be represented under a twirling. Under the framework established in this work, the same operations would instead be correctly identified as non-representable.
These results can help characterise quantum communication schemes, or understand the resources needed to lift a superselection rule, or study implementations of data hiding protocols, and similar issues that are affected by or rely on a lack of knowledge.

\section{Acknowledgements}
The author thanks Stephen Barnett, Daniel Oi and Davide Rinaldi for discussions. This work was supported by the UK EPSRC and the Canada Excellence Research Chairs (CERC) Program.

\bibliography{Twirling}

\begin{thebibliography}{12}%
\makeatletter
\providecommand \@ifxundefined [1]{%
 \@ifx{#1\undefined}
}%
\providecommand \@ifnum [1]{%
 \ifnum #1\expandafter \@firstoftwo
 \else \expandafter \@secondoftwo
 \fi
}%
\providecommand \@ifx [1]{%
 \ifx #1\expandafter \@firstoftwo
 \else \expandafter \@secondoftwo
 \fi
}%
\providecommand \natexlab [1]{#1}%
\providecommand \enquote  [1]{``#1''}%
\providecommand \bibnamefont  [1]{#1}%
\providecommand \bibfnamefont [1]{#1}%
\providecommand \citenamefont [1]{#1}%
\providecommand \href@noop [0]{\@secondoftwo}%
\providecommand \href [0]{\begingroup \@sanitize@url \@href}%
\providecommand \@href[1]{\@@startlink{#1}\@@href}%
\providecommand \@@href[1]{\endgroup#1\@@endlink}%
\providecommand \@sanitize@url [0]{\catcode `\\12\catcode `\$12\catcode
  `\&12\catcode `\#12\catcode `\^12\catcode `\_12\catcode `\%12\relax}%
\providecommand \@@startlink[1]{}%
\providecommand \@@endlink[0]{}%
\providecommand \url  [0]{\begingroup\@sanitize@url \@url }%
\providecommand \@url [1]{\endgroup\@href {#1}{\urlprefix }}%
\providecommand \urlprefix  [0]{URL }%
\providecommand \Eprint [0]{\href }%
\providecommand \doibase [0]{http://dx.doi.org/}%
\providecommand \selectlanguage [0]{\@gobble}%
\providecommand \bibinfo  [0]{\@secondoftwo}%
\providecommand \bibfield  [0]{\@secondoftwo}%
\providecommand \translation [1]{[#1]}%
\providecommand \BibitemOpen [0]{}%
\providecommand \bibitemStop [0]{}%
\providecommand \bibitemNoStop [0]{.\EOS\space}%
\providecommand \EOS [0]{\spacefactor3000\relax}%
\providecommand \BibitemShut  [1]{\csname bibitem#1\endcsname}%
\let\auto@bib@innerbib\@empty
\bibitem [{\citenamefont {Bennett}\ \emph
  {et~al.}(1996{\natexlab{a}})\citenamefont {Bennett}, \citenamefont
  {Brassard}, \citenamefont {Popescu}, \citenamefont {Schumacher},
  \citenamefont {Smolin},\ and\ \citenamefont {Wootters}}]{Bennett1996}%
  \BibitemOpen
  \bibfield  {author} {\bibinfo {author} {\bibfnamefont {C.~H.}\ \bibnamefont
  {Bennett}}, \bibinfo {author} {\bibfnamefont {G.}~\bibnamefont {Brassard}},
  \bibinfo {author} {\bibfnamefont {S.}~\bibnamefont {Popescu}}, \bibinfo
  {author} {\bibfnamefont {B.}~\bibnamefont {Schumacher}}, \bibinfo {author}
  {\bibfnamefont {J.~A.}\ \bibnamefont {Smolin}}, \ and\ \bibinfo {author}
  {\bibfnamefont {W.~K.}\ \bibnamefont {Wootters}},\ }\href {\doibase
  10.1103/PhysRevLett.76.722} {\bibfield  {journal} {\bibinfo  {journal} {Phys.
  Rev. Lett.}\ }\textbf {\bibinfo {volume} {76}},\ \bibinfo {pages} {722}
  (\bibinfo {year} {1996}{\natexlab{a}})}\BibitemShut {NoStop}%
\bibitem [{\citenamefont {Bennett}\ \emph
  {et~al.}(1996{\natexlab{b}})\citenamefont {Bennett}, \citenamefont
  {DiVincenzo}, \citenamefont {Smolin},\ and\ \citenamefont
  {Wootters}}]{Bennett1996b}%
  \BibitemOpen
  \bibfield  {author} {\bibinfo {author} {\bibfnamefont {C.~H.}\ \bibnamefont
  {Bennett}}, \bibinfo {author} {\bibfnamefont {D.~P.}\ \bibnamefont
  {DiVincenzo}}, \bibinfo {author} {\bibfnamefont {J.~A.}\ \bibnamefont
  {Smolin}}, \ and\ \bibinfo {author} {\bibfnamefont {W.~K.}\ \bibnamefont
  {Wootters}},\ }\href {\doibase 10.1103/PhysRevA.54.3824} {\bibfield
  {journal} {\bibinfo  {journal} {Phys. Rev. A}\ }\textbf {\bibinfo {volume}
  {54}},\ \bibinfo {pages} {3824} (\bibinfo {year}
  {1996}{\natexlab{b}})}\BibitemShut {NoStop}%
\bibitem [{\citenamefont {Vollbrecht}\ and\ \citenamefont
  {Werner}(2001)}]{Vollbrecht}%
  \BibitemOpen
  \bibfield  {author} {\bibinfo {author} {\bibfnamefont {K.~G.~H.}\
  \bibnamefont {Vollbrecht}}\ and\ \bibinfo {author} {\bibfnamefont {R.~F.}\
  \bibnamefont {Werner}},\ }\href {\doibase 10.1103/PhysRevA.64.062307}
  {\bibfield  {journal} {\bibinfo  {journal} {Phys. Rev. A}\ }\textbf {\bibinfo
  {volume} {64}},\ \bibinfo {pages} {062307} (\bibinfo {year}
  {2001})}\BibitemShut {NoStop}%
\bibitem [{\citenamefont {Lee}\ \emph {et~al.}(2003)\citenamefont {Lee},
  \citenamefont {Chi}, \citenamefont {Oh},\ and\ \citenamefont {Kim}}]{Lee}%
  \BibitemOpen
  \bibfield  {author} {\bibinfo {author} {\bibfnamefont {S.}~\bibnamefont
  {Lee}}, \bibinfo {author} {\bibfnamefont {D.~P.}\ \bibnamefont {Chi}},
  \bibinfo {author} {\bibfnamefont {S.~D.}\ \bibnamefont {Oh}}, \ and\ \bibinfo
  {author} {\bibfnamefont {J.}~\bibnamefont {Kim}},\ }\href {\doibase
  10.1103/PhysRevA.68.062304} {\bibfield  {journal} {\bibinfo  {journal} {Phys.
  Rev. A}\ }\textbf {\bibinfo {volume} {68}},\ \bibinfo {pages} {062304}
  (\bibinfo {year} {2003})}\BibitemShut {NoStop}%
\bibitem [{\citenamefont {Hayashi}\ \emph {et~al.}(2008)\citenamefont
  {Hayashi}, \citenamefont {Markham}, \citenamefont {Murao}, \citenamefont
  {Owari},\ and\ \citenamefont {Virmani}}]{Hayashi}%
  \BibitemOpen
  \bibfield  {author} {\bibinfo {author} {\bibfnamefont {M.}~\bibnamefont
  {Hayashi}}, \bibinfo {author} {\bibfnamefont {D.}~\bibnamefont {Markham}},
  \bibinfo {author} {\bibfnamefont {M.}~\bibnamefont {Murao}}, \bibinfo
  {author} {\bibfnamefont {M.}~\bibnamefont {Owari}}, \ and\ \bibinfo {author}
  {\bibfnamefont {S.}~\bibnamefont {Virmani}},\ }\href {\doibase
  10.1103/PhysRevA.77.012104} {\bibfield  {journal} {\bibinfo  {journal} {Phys.
  Rev. A}\ }\textbf {\bibinfo {volume} {77}},\ \bibinfo {pages} {012104}
  (\bibinfo {year} {2008})}\BibitemShut {NoStop}%
\bibitem [{\citenamefont {D\"ur}\ \emph {et~al.}(2000)\citenamefont {D\"ur},
  \citenamefont {Cirac}, \citenamefont {Lewenstein},\ and\ \citenamefont
  {Bru\ss{}}}]{Dur}%
  \BibitemOpen
  \bibfield  {author} {\bibinfo {author} {\bibfnamefont {W.}~\bibnamefont
  {D\"ur}}, \bibinfo {author} {\bibfnamefont {J.~I.}\ \bibnamefont {Cirac}},
  \bibinfo {author} {\bibfnamefont {M.}~\bibnamefont {Lewenstein}}, \ and\
  \bibinfo {author} {\bibfnamefont {D.}~\bibnamefont {Bru\ss{}}},\ }\href
  {\doibase 10.1103/PhysRevA.61.062313} {\bibfield  {journal} {\bibinfo
  {journal} {Phys. Rev. A}\ }\textbf {\bibinfo {volume} {61}},\ \bibinfo
  {pages} {062313} (\bibinfo {year} {2000})}\BibitemShut {NoStop}%
\bibitem [{\citenamefont {Gheorghiu}(2012)}]{Vlad2012}%
  \BibitemOpen
  \bibfield  {author} {\bibinfo {author} {\bibfnamefont {V.}~\bibnamefont
  {Gheorghiu}},\ }\href@noop {} {\bibfield  {journal} {\bibinfo  {journal}
  {arXiv:1204.1072 [quant-ph]}\ } (\bibinfo {year} {2012})}\BibitemShut
  {NoStop}%
\bibitem [{\citenamefont {Skotiniotis}\ and\ \citenamefont
  {Gour}(2012)}]{gour2012}%
  \BibitemOpen
  \bibfield  {author} {\bibinfo {author} {\bibfnamefont {M.}~\bibnamefont
  {Skotiniotis}}\ and\ \bibinfo {author} {\bibfnamefont {G.}~\bibnamefont
  {Gour}},\ }\href@noop {} {\bibfield  {journal} {\bibinfo  {journal}
  {arXiv:1202.3163v2 [quant-ph]}\ } (\bibinfo {year} {2012})}\BibitemShut
  {NoStop}%
\bibitem [{\citenamefont {Marvian}\ and\ \citenamefont
  {Spekkens}(2011)}]{marvian}%
  \BibitemOpen
  \bibfield  {author} {\bibinfo {author} {\bibfnamefont {I.}~\bibnamefont
  {Marvian}}\ and\ \bibinfo {author} {\bibfnamefont {R.~W.}\ \bibnamefont
  {Spekkens}},\ }\href@noop {} {\bibfield  {journal} {\bibinfo  {journal}
  {arXiv:1104.0018v1 [quant-ph]}\ } (\bibinfo {year} {2011})}\BibitemShut
  {NoStop}%
\bibitem [{\citenamefont {Bartlett}\ \emph {et~al.}(2009)\citenamefont
  {Bartlett}, \citenamefont {Rudolph}, \citenamefont {Spekkens},\ and\
  \citenamefont {Turner}}]{BRST}%
  \BibitemOpen
  \bibfield  {author} {\bibinfo {author} {\bibfnamefont {S.~D.}\ \bibnamefont
  {Bartlett}}, \bibinfo {author} {\bibfnamefont {T.}~\bibnamefont {Rudolph}},
  \bibinfo {author} {\bibfnamefont {R.~W.}\ \bibnamefont {Spekkens}}, \ and\
  \bibinfo {author} {\bibfnamefont {P.~S.}\ \bibnamefont {Turner}},\
  }\href@noop {} {\bibfield  {journal} {\bibinfo  {journal} {New Journal of
  Physics}\ }\textbf {\bibinfo {volume} {11}},\ \bibinfo {pages} {063013}
  (\bibinfo {year} {2009})}\BibitemShut {NoStop}%
\bibitem [{\citenamefont {Bartlett}\ \emph {et~al.}(2007)\citenamefont
  {Bartlett}, \citenamefont {Rudolph},\ and\ \citenamefont {Spekkens}}]{BRS}%
  \BibitemOpen
  \bibfield  {author} {\bibinfo {author} {\bibfnamefont {S.~D.}\ \bibnamefont
  {Bartlett}}, \bibinfo {author} {\bibfnamefont {T.}~\bibnamefont {Rudolph}}, \
  and\ \bibinfo {author} {\bibfnamefont {R.~W.}\ \bibnamefont {Spekkens}},\
  }\href {\doibase 10.1103/RevModPhys.79.555} {\bibfield  {journal} {\bibinfo
  {journal} {Rev. Mod. Phys.}\ }\textbf {\bibinfo {volume} {79}},\ \bibinfo
  {pages} {555} (\bibinfo {year} {2007})}\BibitemShut {NoStop}%
\bibitem [{\citenamefont {Gour}\ and\ \citenamefont {Spekkens}(2007)}]{GS}%
  \BibitemOpen
  \bibfield  {author} {\bibinfo {author} {\bibfnamefont {G.}~\bibnamefont
  {Gour}}\ and\ \bibinfo {author} {\bibfnamefont {R.~W.}\ \bibnamefont
  {Spekkens}},\ }\href@noop {} {\bibfield  {journal} {\bibinfo  {journal}
  {arXiv:0711.0043v2 [quant-ph]}\ } (\bibinfo {year} {2007})}\BibitemShut
  {NoStop}%
\end{thebibliography}%

\end{document}